\documentclass[aps,twoside]{revtex4}

\usepackage[english]{babel}
\usepackage{graphicx,epic,eepic,epsf,amsmath,latexsym,amssymb,verbatim}
\usepackage{wrapfig,amstext,amsthm}

\newcommand{\bra}[1]{\langle#1|}
\newcommand{\ket}[1]{|#1\rangle}

\newcommand{\ketbra}[2]{|#1\rangle\!\langle#2|}
\newcommand{\pt}{^{\Gamma}}
\newcommand{\eref}[1]{(\ref{#1})}
\def\tr{\mathinner{\mathrm{Tr}}}
\def\rank{\mathinner{\mathrm{rank}}}
\def\ppt{\mathinner{\mathrm{PPT}}}

\def\ox{\otimes}

\def\1{\openone}
\def\L{\left}
\def\R{\right}
\def\dsum{\displaystyle\sum}

\def\PPTto{\stackrel{PPT}{\rightarrow}}

\newtheorem{theorem}{Theorem}
\newtheorem{lemma}[theorem]{Lemma}
\newtheorem{remark}[theorem]{Remark}
\newtheorem{proposition}[theorem]{Proposition}
\newtheorem{corollary}[theorem]{Corollary}
\newtheorem{conjecture}[theorem]{Conjecture}
\newtheorem{example}[theorem]{Example}
\newtheorem{definition}[theorem]{Definition}

\begin{document}
\title{Pure-state transformations and catalysis under operations that completely preserve positivity of partial transpose}
\author{William Matthews}
\email{william.matthews@bris.ac.uk}
\affiliation{Department of Mathematics, University of Bristol, Bristol BS8 1TW, U.K.}
\author{Andreas Winter}
\email{a.j.winter@bris.ac.uk}
\affiliation{Department of Mathematics, University of Bristol, Bristol BS8 1TW, U.K.}
\affiliation{Centre for Quantum Technologies, National University of Singapore, 2 Science Drive 3, Singapore 117542}

\begin{abstract}
Motivated by the desire to better understand the class of quantum operations on bipartite systems that completely preserve positivity of partial transpose (PPT operations) and its relation to the class LOCC (local operations and classical communication), we present some results on deterministic bipartite pure state transformations by PPT operations. Restricting our attention to the case where we start with a rank K maximally entangled state, we give a necessary condition for transforming it into a given pure state, which we show is also sufficient when K is two and the final state has Schmidt rank three. We show that it is sufficient for all K and all final states provided a conjecture about a certain family of semidefinite programs is true. We also demonstrate that the phenomenon of catalysis can occur under PPT operations and that, unlike LOCC catalysis, a maximally entangled state can be a catalyst. Finally, we give a necessary and sufficient condition for the possibility of transforming a rank K maximally entangled state to an arbitrary pure state by PPT operations assisted by some maximally entangled catalyst.
\end{abstract}

\date{28 January 2008}

\maketitle

\section{Introduction}
In the theory of quantum entanglement the familiar `distant labs' scenario motivates the restriction of allowed operations to those that can be carried out by local operations and classical communication (LOCC). Since no entangled state can be created from an unentangled one by LOCC, we regard entanglement as a resource in this context. In the case of the bipartite pure states, the picture of how this resource can be quantified and how it can be transformed (by LOCC) is quite well developed: For bipartite pure states the entanglement of distillation, $E_d$, and the entanglement cost, $E_c$, have been shown to both be equal to the entropy of entanglement (for a review see Plenio and Virmani \cite{emeas}). Since $E_d$ and $E_c$ are equal, LOCC transformation of bipartite pure state entanglement is asymptotically reversible, and the entropy of entanglement is a unique measure of the amount of entanglement available in the many copy limit.

For mixed bipartite states, entanglement transformations are not in general asymptotically reversible, and consequently, LOCC no longer induces a unique measure of the amount of entanglement (in the asymptotic regime). For instance, there are bound entangled states (which by definition have zero entanglement of distillation) which have been shown to have non-zero entanglement cost \cite{irrev}. This has motivated the search for non-trivial extensions of LOCC, with respect to which entanglement transformations are asymptotically reversible for all bipartite states. Recently, Plenio and Brand\~{a}o \cite{pb} proved that the set of operations which asymptotically cannot generate entanglement is such a set, with the regularised relative entropy of entanglement as the corresponding measure of entanglement. An open conjecture \cite{ape} is that the set of PPT operations \cite{rains} (defined later in this section) also renders bipartite entanglement transformations asymptotically reversible (in \cite{ape} it is shown that it does so for the anti-symmetric Werner state).

For exact, finite transformations, the problem of determining when a given pure bipartite state can be deterministically converted into another by LOCC is completely solved by Nielsen's majorization theorem \cite{nielsen}: The process is possible if and only if the Schmidt coefficient vector of the initial state is majorized by that of the final state.

In the present paper we investigate exact, finite transformations of pure states by PPT operations. This topic was first treated by Ishizaka and Plenio \cite{IshizakaPlenio}), where the emphasis is on conversion of LOCC-inequivalent forms of multi-partite entanglement (e.g. GHZ and W states for three parties) by PPT operations. In the present work we are only concerned with bipartite states.

Here, we introduce some conventions that will be used throughout the rest of the paper. Logarithms are always taken to base two. If the deterministic transformation of a state $\rho$ into $\rho'$ can be accomplished by the class of operations $OP$ then we write $\rho \stackrel{OP}{\rightarrow} \rho'$; if it cannot be then we write $\rho \stackrel{OP}{\nrightarrow} \rho'$.  In discussing the transformation of bipartite pure states by any class of operations which contains LOCC, we need only consider the Schmidt coefficients of the states since states with the same Schmidt coefficients are equivalent up to local unitary transformations (which are obviously contained in LOCC). Since we are only concerned here with classes that include LOCC, we will use the state

\begin{equation}
	\rho_{\lambda} = \sum_{i=1}^d\sum_{j=1}^d \sqrt{\lambda_{i}\lambda_{j}} \ket{ii}\bra{jj}
\end{equation}

(where $\ket{ij} = \ket{i}_A \ox \ket{j}_B$ for orthonormal bases $\{\ket{i}_A\}$, $\{\ket{i}_B\}$ for Alice and Bob's quotients of the bipartite Hilbert space) as a representative of all states with Schmidt coefficient vector $\lambda$ without loss of generality. We use $\lambda^{\uparrow}$ ($\lambda^{\downarrow}$) to denote the vector obtained by putting the components of $\lambda$ in non-decreasing (non-increasing) order.

Using this notation, Nielsen's theorem is
\begin{equation}
	\rho_{\lambda} \stackrel{LOCC}{\rightarrow} \rho_{\mu} \iff \lambda \prec \mu,
\end{equation}
where the majorization relation is defined on vectors in $\mathbb{R}^d$ whose components sum to one by
\begin{equation}
	\lambda \prec \mu \iff \sum_{i=1}^{j}\lambda_i^{\downarrow} \leq \sum_{i=1}^{j}\mu_i^{\downarrow}, \textrm{ for all } j \in \{1,...,d\}.
\end{equation}
We shall use $\Phi_K$ to denote a maximally entangled state of rank $K$ (where we assume that $K \geq 2$).

We will use the symbol $X\pt$ to denote the partial transpose an operator $X$. We define the (linear) partial transpose map by
\begin{equation}
	\ketbra{ij}{kl}\pt = \ketbra{il}{kj}.
\end{equation}
Clearly this is basis dependent, but the eigenvalues (and hence the positivity) of the partial transpose of an operator does not depend on this basis choice. For convenience we will use the same basis that we are using as our representative Schmidt basis.

The set of PPT operations on a bipartite system is the set of completely positive trace-preserving (CPTP) maps $\mathcal{L}$ such that the composition $\Gamma \circ \mathcal{L} \circ\Gamma$ is also completely positive, where $\Gamma$ denotes the partial transpose map $\rho \to \rho\pt$. Equivalently, PPT maps are CPTP maps that \emph{completely} preserve the PPT property of states in the same sense that completely positive maps \emph{completely} preserve non-negativity of states: Any extension of a PPT map on a system Q onto a larger system QR where we apply the original map on Q and the identity map on R is PPT preserving (it maps PPT states to PPT states).

We make frequent use of the R\'{e}nyi entropies: For $t \in [0,\infty]$ the \emph{R\'{e}nyi entropy at $t$} is defined by

\begin{equation}
		S_t\left(\lambda\right) = \left\{ \begin{array}{ll}
			\frac{1}{1-t}\log \left( \sum_{i=1}^{d} \lambda_i^t \right) &\textrm{ for } t \in (0,1)\cup(1,\infty), \\
			\log |\lambda| &\textrm{ for } t = 0,\\
			H(\lambda) &\textrm{ for } t = 1,\\
			-\log \lambda^{\downarrow}_1 &\textrm{ for } t = \infty,
		\end{array} \right.
\end{equation}
where $|\lambda|$ denotes the number of non-zero components of $\lambda$ (i.e. the Schmidt rank) and $H(\lambda) = -\sum_i \lambda_i\log\lambda_i$ is the Shannon entropy of $\lambda$.

In the next section we discuss some necessary conditions for general bipartite pure state transformations by PPT. After that we provide part of the theory analogous to Nielsen's theory for deterministic transformations of pure bipartite states by PPT operations. In particular, we look at the case where the initial state is maximally entangled. In section \ref{sec:frommax} we provide a necessary condition for transformations of this type, which we conjecture is also sufficient. We make use of this result in section \ref{sec:cat} to show that the phenomenon of catalysis \cite{catalysis} can occur for PPT operations and give a necessary and sufficient condition for when this can occur if both the initial state and catalyst state are maximally entangled. We conclude with some suggestions for future work, including a conjecture for extending this condition to arbitrary initial and catalyst states.

\section{General PPT pure state transformations}\label{sec:general}

Determining whether a particular pure state transformation can be carried out by a PPT operation can in fact be formulated as a semidefinite programming problem \cite{IshizakaPlenio}, the difficulty is in phrasing the relevant constraints in terms of conditions on the Schmidt coefficients of the pure states. Here we show that certain R\'{e}nyi entropies of the Schmidt coefficient vectors correspond to operationally motivated PPT monotones (allowing us to give some necessary conditions for transformations).

For completeness let us define these operational monotones for any class of operations $X$: Let $\rho$ be a density operator in $\mathcal{B}(\mathcal{H}_A\ox\mathcal{H}_B)$, where $\mathcal{B}(\mathcal{H})$ denotes the set of hermitian operators on the Hilbert space $\mathcal{H}$.
\begin{definition}
	The distillable entanglement of $\rho$ under the class of operations $X$ is defined by
	\begin{equation}
		E_{d}^{{\rm X}}( \rho ) := \sup_{(\mathcal{L}_j)_{j\in\mathbb{N}}}\L\{ \lim_{j\to\infty} \frac{\log K_j}{n_j} \bigg| \lim_{j\to\infty} \| \Phi_{K_j} - \mathcal{L}_j ( \rho^{\ox n_j} ) \|_1 = 0 \R\},
	\end{equation}
	and the exact distillable entanglement by
	\begin{equation}
		E_{dx}^{{\rm X}}( \rho ) := \lim_{j\to\infty} \max_{\mathcal{L}_j}\L\{ \frac{\log K_j}{n_j} \bigg| \| \Phi_{K_j} - \mathcal{L}_j ( \rho^{\ox n_j} ) \|_1 = 0 \R\},
	\end{equation}
	where in both cases $(\mathcal{L}_j)_{j\in\mathbb{N}}$ is a sequence of completely positive trace preserving (CPTP) maps such that each map, $\mathcal{L}_j$, takes $\mathcal{B}\L(\L(\mathcal{H}_A\ox\mathcal{H}_B\R)^{\ox n_j}\R)$ to $\mathcal{B}\L(\mathbb{C}^{\ox K_j}\ox\mathbb{C}^{\ox K_j}\R)$ and belongs to the class $X$.
\end{definition}

\begin{definition}
	The entanglement cost of a state $\rho$, we define by
	\begin{equation}
		E_{c}^{{\rm X}}( \rho ) := \inf_{(\mathcal{L}_j)_{j\in\mathbb{N}}}\L\{ \lim_{j\to\infty} \frac{\log K_j}{n_j} \bigg| \lim_{j\to\infty} \| \mathcal{L}_j (\Phi_{K_j} ) -  \rho^{\ox n_j} \|_1 = 0 \R\},
	\end{equation}
	and the exact entanglement cost by
	\begin{equation}
		E_{cx}^{{\rm X}}( \rho ) := \lim_{j\to\infty} \min_{\mathcal{L}_j}\L\{ \frac{\log K_j}{n_j} \bigg| \| \mathcal{L}_j (\Phi_{K_j} ) -  \rho^{\ox n_j} \|_1 = 0 \R\},
	\end{equation}
	where in both cases $(\mathcal{L}_j)_{j\in\mathbb{N}}$ is a sequence of CPTP maps such that each map, $\mathcal{L}_j$, takes $\mathcal{B}\L(\mathbb{C}^{\ox K_j}\ox\mathbb{C}^{\ox K_j}\R)$ to $\mathcal{B}\L(\L(\mathcal{H}_A\ox\mathcal{H}_B)\R)^{\ox n_j}\R)$ and belongs to the class $X$.
\end{definition}

\begin{proposition}
	The entanglement cost $E_{c}^{\ppt}(\rho_{\lambda})$ and distillable entanglement $E_{d}^{\ppt}(\rho_{\lambda})$ of $\rho_{\lambda}$ are both equal to $S_{1}(\lambda)$, the entropy of entanglement of the state.
\end{proposition}
\begin{proof}
	
It is clear and well-known that by LOCC operations, both the entanglement cost and distillable entanglement of $\rho_{\lambda}$ is $S_{1}(\lambda)$. By elementary
results of Rains' theory of PPT distillation \cite{rains}, PPT operations
cannot asymptotically increase the number of EPR pairs. Hence
$$E_{c}^{\ppt}(\rho_{\lambda}) \geq S_1(\lambda),\quad
  E_{d}^{\ppt}(\rho_{\lambda}) \leq S_1(\lambda).$$
Since the opposite inequalities are trivial by LOCC $\subset$ PPT,
we conclude that
$$E_{c}^{\ppt}(\rho_{\lambda}) = E_{d}^{\ppt}(\rho_{\lambda}) = S_{1}(\lambda),$$
thus coinciding with the corresponding LOCC entanglement cost and distillable entanglement: the entropy
of entanglement.
\end{proof}

\begin{proposition}
	The exact distillable entanglement $E_{xd}^{\ppt}(\rho_{\lambda})$ of $\rho_{\lambda}$ is given by $S_{\infty}(\lambda)$.
\end{proposition}
\begin{proof}
We use a result of Rains~\cite{rains} on the maximum fidelity
obtainable from $\ket{\psi}$ via PPT operations, to a maximally entangled
state of Schmidt rank $K$: by this result, exact transformation (fidelity $1$)
is possible if there exists an operator $F$ with
$$\rho_{\lambda} \leq F\leq\1,\quad -\frac{1}{K}\1\leq F^\Gamma\leq \frac{1}{K}\1.$$
Thus, an upper bound to $K$ is given by
$$\max\left\{ \left\|F^\Gamma\right\|_\infty^{-1}\ \big|\  \rho_{\lambda}\leq F\leq\1 \right\}.$$
We claim that this is $1/\lambda^{\downarrow}_1$: for assume any
$F$ as above, and write it $F=\rho_{\lambda^{\downarrow}}+A$, $A\geq 0$.
\begin{equation}\begin{split}
  \left\|F^\Gamma\right\|_\infty &\geq    \tr(\ketbra{11}{11}(\rho_{\lambda^{\downarrow}}\pt + A\pt)) \\
                               &=    \lambda^{\downarrow}_1 + \tr(\ketbra{11}{11}\pt A) \\
& = \lambda^{\downarrow}_1 + \tr(\ketbra{11}{11} A) \geq \lambda^{\downarrow}_1.
\end{split}\end{equation}
Hence $\left\|F^\Gamma\right\|_\infty\geq \lambda^{\downarrow}_1$, with equality
achieved for $F=\rho_{\lambda}$.
\par
This gives $E_{xd}^{\ppt}(\rho_{\lambda}) \leq S_\infty(\lambda)$.
The opposite inequality comes from LOCC operations asymptotically
achieving this bound~\cite{morikoshi}.
\end{proof}

\begin{proposition}\label{exactPPTcost}
The entanglement cost $E_{xc}^{\ppt}(\rho_{\lambda})$ of $\rho_{\lambda}$ is $S_{1/2}(\lambda)$.
\end{proposition}
\begin{proof}
	This is merely an application of the more general result of Audenaert \emph{et al.} \cite{ape} to pure states.
\end{proof}

All of these quantities are clearly PPT monotones so they provide a necessary condition for the possibility of a PPT pure state transformation:
\begin{proposition}
	If $\rho_{\mu} \PPTto \rho_{\lambda}$ then
	\begin{align}
		S_{1/2}(\mu) \geq S_{1/2}(\lambda), \\
		S_{1}(\mu) \geq S_{1}(\lambda), \\
		S_{\infty}(\mu) \geq S_{\infty}(\lambda).
	\end{align}
\end{proposition}

As Schur concave functions, the R\'{e}nyi entropies of the Schmidt coefficients at all values of $t \in [0,\infty]$ are monotones for LOCC state transformations. Under PPT however, the R\'{e}nyi entropies for $0 \leq t < 1/2$ are \emph{not} monotones.

\begin{example}
  \label{expl:small-alpha}
  Consider a pure state $\rho_{\lambda}$ with Schmidt spectrum
  $$\lambda^{\downarrow} = \left( \frac{1}{20},\frac{1}{20},\frac{1}{20},\frac{4}{20},\frac{4}{20},\frac{9}{20} \right).$$
 \par
  It is easily verified that $S_{1/2}(\lambda)=\log 5$, and
  indeed, in accordance with the Proposition \ref{exactPPTcost} the transformation
  $$\Phi_5^{\otimes(n+o(n))}\longrightarrow\rho_{\lambda}^{\otimes n}$$
  is possible by PPT operations for all sufficiently large $n$.
  \par
  But for $0\leq \alpha<1/2$,
  $$\log 5 < S_{t}(\lambda).$$
  Since there are trivial examples of initial states such that the opposite
  inequality is true (e.g.~$\Phi_6$), the R\'{e}nyi entropies
  at $0\leq t<1/2$ are not PPT monotones.
  \par
  $S_0(\lambda)$ is just the logarithm of the Schmidt rank, so, as was noted in \cite{IshizakaPlenio}, PPT operations can increase the Schmidt rank of pure states, a thing LOCC transformations cannot even do with nonzero probability!
\end{example}

It should be noted that a necessary condition for pure bipartite state transformations by \emph{separable} operations was recently given by Gheorghiou and Griffiths \cite{sepTrans}.

\section{PPT transformations from maximally entangled states.}\label{sec:frommax}

Unless otherwise stated, the final state is $\rho_{\lambda}$, where the Schmidt coefficient vector $\lambda \in \mathbb{R}^d$ is assumed without loss of generality to have no vanishing components (so $d$ is the Schmidt rank of the final state). Since any state with Schmidt rank not greater than $K$ can be produced from $\Phi_K$ by LOCC, the interesting case for PPT transformations is when the Schmidt rank is increased.

\begin{lemma}
	For any (pure or mixed) final state $\rho$, $\Phi_K \stackrel{PPT}{\rightarrow} \rho$ if and only if the solution to the semidefinite program
\begin{equation}\label{mixedPrimalSDP}
	\min \{ \tr \left(P\right) | P \ge 0, -\left(K-1\right)P\pt \leq \rho\pt \leq \left(K+1\right)P\pt \},
\end{equation}
	where $P$ is an hermitian operator on the same Hilbert space as $\rho$, is less than or equal to one.
\end{lemma}
\begin{proof}
	The argument is almost the same as the one given in \cite{ape}, but we present it here for convenience.
	If there is a PPT map $\mathcal{L}$ such that $\mathcal{L}\left(\Phi_K\right) = \rho$ then the map $\mathcal{L'}$ made by preceding $\mathcal{L}$ with the twirl operation \cite{rains} $\mathcal{T}$ (which can be implemented by LOCC) is also PPT, and does the same transformation since $\mathcal{T}\left(\Phi_K\right) = \Phi_K$. By symmetry it is always possible to write the new map in the form
	\begin{equation}
		\mathcal{L'}\left(\tau\right) = F \tr\left(\Phi_K \tau\right) + G \tr\left(\left(\1 - \Phi_K\right) \tau\right)
	\end{equation}
	where, in order for the map to be CPTP, $F$ and $G$ must be density operators. In order that $\mathcal{L}'\left(\tau\right) = \rho$, we require $F = \rho$, so the desired PPT transformation is possible if and only if there is a density operator $G$ such that the resulting $\mathcal{L'}$ is PPT. $\mathcal{L'}$ is PPT preserving if and only if the operator
	\begin{align}
		\mathcal{L'}\left(\tau \pt\right)\pt &= \rho\pt \tr\left(\Phi_K\pt \tau\right) + G\pt \tr\left(\left(\1 - \Phi_K\right)\pt \tau\right)\\
		&= \rho\pt \tr\left(\left(\mathcal{S} - \mathcal{A}\right) \tau\right)/K + G\pt \left(\left(\left(1-1/K\right)\mathcal{S} + \left(1+1/K\right)\mathcal{A}\right) \tau\right) \\
		&= \tr\left(\mathcal{S} \tau\right) \left(\rho\pt/K + \left(1-1/K\right)G\pt\right) + \tr\left(\mathcal{A} \tau\right) \left(\left(1+1/K\right)G\pt - \rho\pt/K\right)
	\end{align}
	is positive semidefinite for all positive semidefinite $\tau$, where $\mathcal{S}$ and $\mathcal{A}$ are the projectors onto the symmetric and anti-symmetric subspaces of the bipartite space. Since $\mathcal{S}\mathcal{A} = 0$, this condition holds if and only if $\rho\pt/K + \left(1-1/K\right)G\pt \geq 0$ and $\left(1+1/K\right)G\pt - \rho\pt/K \geq 0$. We also require $G \geq 0$ and $\tr G = 1$. The set of operators which satisfy these constraints is precisely the set of feasible points of the SDP \eref{mixedPrimalSDP} which have trace 1. If $P_0$ is feasible, then points $P_0 + t \1$ are also feasible for all $t \geq 0$, so if an optimal point has trace not greater than one, it ensures the existence of a feasible point which satisfies all the constraints on $G$ (so the transformation is possible); if an optimal point has trace greater than one, then clearly no such point exists and hence the transformation is not possible.
	
\end{proof}
	
\begin{proposition}\label{prop:primal}
	\begin{equation}
		\Phi_K \stackrel{PPT}{\rightarrow} \rho_{\lambda} \iff T\left(K;\lambda\right) \leq 1,
	\end{equation}
	where we define $T\left(K;\lambda\right)$ to be the solution to the semidefinite program
	\begin{equation}\label{primalSDP}
		T\left(K;\lambda\right) := \min \left\{ \dsum_{i\geq j}^{d} s_{ij} + \dsum_{i > j}^{d} a_{ij} \bigg| s_{ij} \geq \frac{\sqrt{\lambda_i \lambda_j}}{K+1}, a_{ij} \geq \frac{\sqrt{\lambda_i \lambda_j}}{K-1}, \dsum_{i\geq j} s_{ij} \sigma_{ij}\pt + \dsum_{i > j} a_{ij} \alpha_{ij}\pt \geq 0 \right\},
	\end{equation}
	with
	\begin{eqnarray*}
	&\sigma_{i j}& = \left\{ \begin{array}{l}
		(\ket{ij} + \ket{ji}) (\bra{ij} + \bra{ji})/2 \\
		\ketbra{ii}{ii}
		  \end{array} \right. \begin{array}{c}
		    \textrm{ when } i \neq j, \\
		    \textrm{ when } i = j,
		  \end{array}\\
		  &\alpha_{i j}& = \left(\ket{ij} - \ket{ji}\right) \left(\bra{ij} - \bra{ji}\right)/2,
	\end{eqnarray*}
	and $\{s_{ij} | 1 \leq j \leq i \leq d \}$, $\{a_{ij} | 1 \leq j < i \leq d \}$ together constitute $d^2$ real variables.
\end{proposition}
\begin{proof}
\begin{equation}
	\rho_{\lambda}\pt = \displaystyle\sum_{i\geq j}\sqrt{\lambda_i \lambda_j} \sigma_{ij} - \displaystyle\sum_{i > j}\sqrt{\lambda_i \lambda_j} \alpha_{ij}.
\end{equation}
Let $\Pi$ be the projection map on the space of hermitian operators given by
\begin{equation}
	\Pi\left(\tau\right) = \displaystyle\sum_{i\geq j}\sigma_{ij}\tau\sigma_{ij} + \displaystyle\sum_{i > j}\alpha_{ij}\tau\alpha_{ij}.
\end{equation}
To show that this map preserves positivity of partial-transpose, we note that
\begin{equation}
	\Pi\left(\tau\pt\right)\pt = \frac{1}{2}\displaystyle\sum_{i\neq j}\bra{ij}\tau\ket{ij}\left(\ketbra{ij}{ij} + \ketbra{ji}{ji}\right) + \frac{1}{2}\displaystyle\left(\sum_{i}\ketbra{ii}{ii}\right)\left(\tau+\tau^{\ast}\right)\left(\sum_{j}\ketbra{jj}{jj}\right).
\end{equation}
Clearly then, $\Pi$ is positive and preserves positivity of partial transpose. Since $\rho\pt$ lies in the image of $\Pi$, if $P$ is a feasible point of the semidefinite program \eref{mixedPrimalSDP} then $\Pi\left(P\pt\right)\pt$ is also a feasible point. Since $\Pi$ is also trace reducing, it will not change the optimal value of \eref{mixedPrimalSDP} if we impose the additional restriction that $P = \Pi\left(P\pt\right)\pt$. This restriction is obeyed if and only if we can write $P$ in the form
\begin{equation}
	P = \dsum_{i\geq j} s_{ij} \sigma_{ij}\pt + \dsum_{i > j} a_{ij} \alpha_{ij}\pt.
\end{equation}
Substituting this into \eref{mixedPrimalSDP}, we obtain \eref{primalSDP}.
\end{proof}

\begin{lemma}
	\begin{equation}\label{T-upperBound}
		T\left(K; \lambda\right) \leq \left(2^{S_{1/2}\left(\lambda\right)}-1\right)/\left(K-1\right).
	\end{equation}
\end{lemma}
\begin{proof}
	The point
	\begin{equation}
		s_{ij} = \sqrt{\lambda_i \lambda_j}/\left(K-1\right), a_{ij} = \sqrt{\lambda_i\lambda_j}/\left(K-1\right)
	\end{equation}
	is primal feasible since
	\begin{equation}
		\dsum_{i} \lambda_i \sigma_{ii}\pt + \dsum_{i > j} \sqrt{\lambda_i\lambda_j} \left(\sigma_{ij} + \alpha_{ij}\right)\pt = \dsum_{i} \lambda_i \ketbra{ii}{ii} + \dsum_{i > j} \sqrt{\lambda_i\lambda_j} \left(\ketbra{ij}{ij} + \ketbra{ji}{ji}\right)\geq 0
	\end{equation}
	and the other inequalities are obviously satisfied. The primal objective function at this point is $2^{S_{1/2}\left(\lambda\right)}/\left(K-1\right)$.
\end{proof}

Since the semidefinite program \eref{primalSDP} is strictly feasible (take the point $s_{ij} = 2, a_{ij} = 2$, for example), its solution is equal to the solution of the dual SDP \cite{vbSDP} so we have

\begin{proposition}
	\begin{align}
		T\left(K;\lambda\right) = \max \left\{ \frac{K 2^{S_{1/2}\left(\lambda\right)} - 1}{K^2 - 1} - \left(\sum_{i\geq j} \frac{\mu_{ij}\sqrt{\lambda_i \lambda_j}}{K+1} + \sum_{i > j} \frac{t_{ij}\sqrt{\lambda_i \lambda_j}}{K-1} \right) \bigg| \left(\mu,t\right) \in R \right\}\label{dualSDP}\\
		R := \left\{\left(\mu,t\right) \bigg| \mu_{ij} \leq 1, t_{ij} \leq 1, \sum_{i\geq j} \mu_{ij} \sigma_{ij}\pt + \sum_{i > j} t_{ij} \alpha_{ij}\pt \geq 0 \right\}.
	\end{align}
	Here $\{\mu_{ij} | 1 \leq j \leq i \leq d \}$, $\{t_{ij} | 1 \leq j < i \leq d \}$ together constitute $d^2$ real variables.
\end{proposition}
\begin{proof}
	The semidefinite program dual to \eref{primalSDP} is
	\begin{align}
		\textrm{maximize } \frac{1}{K+1}\sum_{i\geq j}\sqrt{\lambda_i \lambda_j} \left(1-\tr\left(Z \sigma_{ij}\pt\right)\right) + \frac{1}{K-1}\sum_{i > j}\sqrt{\lambda_i \lambda_j} \left(1-\tr\left(Z \alpha_{ij}\pt\right)\right)
	\end{align}
	\begin{align}	
		\textrm{subject to } Z &\geq 0\\
		\tr\left(Z \sigma_{ij}\pt\right) &\leq 1\\
		\tr\left(Z \alpha_{ij}\pt\right) &\leq 1.
	\end{align}
	If $Z$ is feasible point of this program then so is $\Pi\left(Z\pt\right)\pt$ (where $\Pi$ is the map defined in the proof of Proposition \ref{prop:primal}). So the substitutions
	\begin{equation}
		Z = \sum_{i\geq j} \mu_{ij} \sigma_{ij}\pt + \sum_{i > j} t_{ij} \alpha_{ij}\pt, x_{ij} = 1 - \mu_{ij}, y_{ij} = 1 - t_{ij}
	\end{equation}
	result in an SDP with the same solution and this is the one given in the proposition.
\end{proof}

The dual objective at any feasible point of the dual semidefinite program \eref{dualSDP} is a lower bound on $T\left(K;\lambda\right)$, and therefore

\begin{corollary}
	\begin{equation}\label{T-lowerBound}
		T\left(K; \lambda\right) \geq \left(K 2^{S_{1/2}\left(\lambda\right)} - 1\right)/\left(K^2 - 1\right).
	\end{equation}
\end{corollary}
\begin{proof}
	The point $\mu_{ij} = 0, t_{ij} = 0$ is clearly dual feasible.
\end{proof}

\begin{theorem}\label{t:rank1opt}
	The optimal value of the dual objective that can be attained by a dual feasible point which satisfies the additional constraint that
	\begin{equation}\label{rankconstraint}
		\rank\left(\dsum_{i\geq j} \mu_{ij} \sigma_{ij}\pt + \dsum_{i > j} t_{ij} \alpha_{ij}\pt\right) = 1,
	\end{equation}
 	is given by
	\begin{equation}\label{rankConstrainedOptimum}
		T_1\left(K; \lambda\right) = \frac{K 2^{S_{1/2}\left(\lambda\right)} - 1}{K^2 - 1}+\frac{K}{\left(K^2-1\right)\left(K+c^{\ast}-d\right)}\left( \left(\dsum_{i=1}^{c^{\ast}}\sqrt{\lambda^{\uparrow}_i}\right)^2-\left(K+c^{\ast}-d\right)\dsum_{i=1}^{c^{\ast}}\lambda^{\uparrow}_i\right),
	\end{equation}
	where $c^{\ast}$ is the smallest number $c \in \{1 + d - K,...,d-1 \}$ satisfying
	\begin{equation}\label{c-constraint}
		\frac{\sum_{i=1}^{c}\sqrt{\lambda^{\uparrow}_i}}{K+c-d} \leq \sqrt{\lambda^{\uparrow}_{c+1}}.
	\end{equation}
	If none of the integers in the range satisfy this relation then $c^{\ast} = d$.
\end{theorem}
(The proof of this theorem is given in the appendix.)
\begin{remark}
	Clearly $T_1\left(K; \lambda\right) \leq T\left(K; \lambda\right)$, so a necessary condition for the transformation $\Phi_K \PPTto \rho_{\lambda}$ is $T_1\left(K; \lambda\right) \leq 1$.
\end{remark}

\begin{corollary}\label{c:tt1}
	If $\left(\dsum_{i=1}^{d-1}\sqrt{\lambda^{\uparrow}_i}\right)/\left(K-1\right) > \sqrt{\lambda^{\uparrow}_{d}}$ then
	\begin{equation}
		T\left(K;\lambda\right) = T_1\left(K;\lambda\right) = \left(2^{S_{1/2}\left(\lambda\right)}-1\right)/\left(K-1\right).
	\end{equation}
\end{corollary}
\begin{proof}
	In this case $c^{\ast} = d$ and so, $T_1\left(K;\lambda\right) = \left(2^{S_{1/2}\left(\lambda\right)}-1\right)/\left(K-1\right)$ which, by Lemma \ref{T-upperBound}, is an upper bound on $T\left(K;\lambda\right)$. Since $T_1$ is also a lower bound on $T$, the result follows.
\end{proof}

\begin{corollary}\label{c:borderline}
	If $S_{1/2}\left(\lambda\right) = \log K$ and $d \geq K$ then the transformation is possible only in the trivial case where the goal state is also a maximally entangled state of rank $K$.
\end{corollary}
\begin{proof}
	The value $c = d - 1$ satisfies \eref{c-constraint} provided that $\lambda^{\uparrow}_d \geq 1/K$ or equivalently if $S_{\infty}\left(\lambda\right) \leq \log K$. Since $S_{t}\left(\lambda\right)$ is a non-increasing function of $t$, this condition is indeed satisfied. Using this value of $c$ in \eref{rankConstrainedOptimum} yields the lower bound
	\begin{equation}
		T_1\left(K; \lambda\right) \geq 1 + \frac{K\left(\sqrt{K \lambda^{\uparrow}_d} - 1\right)^2}{\left(K^2-1\right)\left(K-1\right)}
	\end{equation}
	which is clearly greater than $1$ (implying the impossibility of the transformation) except where $S_{\infty}\left(\lambda\right) = \log K $ which (together with $S_{1/2}\left(\lambda\right) = \log K$) implies $\lambda$ is the uniform distribution of size $K$, so the goal state is a maximally entangled state of rank $K$.
\end{proof}

\begin{proposition}\label{p:rank3}
	In the case where the goal state has Schmidt rank three, $T\left(2; \lambda\right) = T_1\left(2; \lambda\right)$.
\end{proposition}
\begin{proof}
	To simplify notation we shall here assume that $\lambda_i = \lambda^{\uparrow}_i$. If $\sqrt{\lambda_1} + \sqrt{\lambda_2} > \sqrt{\lambda_3}$ then we can apply Corollary \ref{c:tt1} and we're done, so we assume that $\sqrt{\lambda_1} + \sqrt{\lambda_2} \leq \sqrt{\lambda_3}$. In this case
	\begin{equation}
		T_1 = \left(1 + 4\left(\sqrt{\lambda_3 \lambda_2} + \sqrt{\lambda_3 \lambda_1}\right) + 8\sqrt{\lambda_1 \lambda_2 }\right)/3.
	\end{equation}
	We shall show that $T$ is the same by constructing a primal optimal solution. Let
	\begin{eqnarray}
		s^{\ast}_{11} = \frac{1}{3}\lambda_1 + \frac{4}{9}\sqrt{\lambda_1 \lambda_2}, s^{\ast}_{22} = \frac{1}{3}\lambda_2 + \frac{4}{9}\sqrt{\lambda_1 \lambda_2}, s^{\ast}_{33} = \frac{1}{3}\lambda_3, \\
		s^{\ast}_{12} = \left(\frac{1}{3} + \frac{4}{9}\right)\sqrt{\lambda_1 \lambda_2}, s^{\ast}_{13} = \frac{1}{3}\sqrt{\lambda_1 \lambda_3}, s^{\ast}_{23} = \frac{1}{3}\sqrt{\lambda_2 \lambda_3}, \\
		a^{\ast}_{ij} = \sqrt{\lambda_i \lambda_j}
	\end{eqnarray}
	\begin{equation}
	\dsum_{i\geq j} s^{\ast}_{ij} \sigma_{ij}\pt + \dsum_{i > j} a^{\ast}_{ij} \alpha_{ij}\pt = P_1\oplus P_2
	\end{equation}
	\begin{eqnarray*}
	  P_1 = \frac{2}{3} \sum_{i \neq j} \sqrt{\lambda_i
	  \lambda_j} \ketbra{ij}{ij} + \frac{2}{9} \left(\ketbra{23}{23} + \ketbra{32}{32}\right) \geq 0
	\end{eqnarray*}
	Written as a matrix in the $\{\ketbra{ii}{jj}\}$ basis,
	\begin{eqnarray*}
	  P_2 = \frac{1}{9} \left(\begin{array}{ccc}
	    3 \lambda_1 + 4 \sqrt{\lambda_2 \lambda_3} & - 3 \sqrt{\lambda_1 \lambda_2} & - 3 \sqrt{\lambda_1
	    \lambda_3}\\
	    - 3 \sqrt{\lambda_2 \lambda_1} & 3 \lambda_2 + 4 \sqrt{\lambda_2
	    \lambda_3} & - \sqrt{\lambda_2 \lambda_3}\\
	    - 3 \sqrt{\lambda_3 \lambda_1} & - \sqrt{\lambda_3 \lambda_2} & 3
	    \lambda_3
	  \end{array}\right) &  & 
	\end{eqnarray*}
	which can be seen to be positive semidefinite by the Sylvester criterion. Therefore $P_1 \oplus P_2 \geq 0$, and since the other primal constraints are clearly satisfied, the point $\left(s^{\ast},a^{\ast}\right)$ is primal feasible so
	\begin{equation}
		T_1\left(2;\lambda\right) \leq T\left(2;\lambda\right) \leq \dsum_{i\geq j} s^{\ast}_{ij} + \dsum_{i > j} a^{\ast}_{ij} = \left(1 + 4\left(\sqrt{\lambda_3 \lambda_2} + \sqrt{\lambda_3 \lambda_1}\right) + 8\sqrt{\lambda_1 \lambda_2 }\right)/3 = T_1\left(2;\lambda\right).
	\end{equation}
	Therefore $T\left(2;\lambda\right) = T_1\left(2;\lambda\right)$.
\end{proof}
It follows directly from this proposition that:
\begin{theorem}\label{t:rank3}
	A pure state of Schmidt rank three can be produced from $\Phi_2$ (an EPR pair) by PPT operations if and only if its Schmidt coefficients obey
	\begin{equation}
		2\left( \sqrt{\lambda^{\uparrow}_3 \lambda^{\uparrow}_2} + \sqrt{\lambda^{\uparrow}_3 \lambda^{\uparrow}_1} \right) + 4\sqrt{\lambda^{\uparrow}_1\lambda^{\uparrow}_2} \leq 1.
	\end{equation}
\end{theorem}

This fact, and some suggestive numerical evidence, leads us to make the following conjecture
\begin{conjecture}\label{conj:rankIs1}
	In the dual program \eref{dualSDP}, the optimal value of the dual objective function is always attained by a point satisfying the rank constraint \eref{rankconstraint} and as a consequence $T\left(K; \lambda\right) = T_1\left(K; \lambda\right)$.
\end{conjecture}

\section{Catalysis}\label{sec:cat}

Nielsen's theorem was used to show that a phenomenon analogous to chemical catalysis can occur in LOCC entanglement transformation \cite{catalysis}. That is, there exist pairs of states $\rho_{\lambda}, \rho_{\mu}$ such that $\rho_{\lambda} \stackrel{LOCC}{\nleftrightarrow} \rho_{\mu}$ but where there is a third state $\rho_{\xi}$ such that $\rho_{\lambda}\ox\rho_{\xi} \stackrel{LOCC}{\rightarrow} \rho_{\mu}\ox\rho_{\xi}$.

A necessary and sufficient condition for the existence of a deterministic catalytic transformation for bipartite pure states was recently given in papers by Turgut \cite{turgut} and Klimesh \cite{klimesh}:

\begin{theorem}\label{t:locccatcond}
Given two pure bipartite states $\rho_{\lambda}$ and $\rho_{\mu}$, where $\lambda^{\uparrow} \neq \mu^{\uparrow}$ and $\lambda$ and $\mu$ don't both have components equal to zero, there exists a pure state $\rho_{\xi}$ such that $\rho_{\lambda}\ox\rho_{\xi} \stackrel{LOCC}{\rightarrow} \rho_{\mu}\ox\rho_{\xi}$ if and only if the following conditions are satisfied
\begin{align}
	S_{t}\left(\lambda\right) &> S_{t}\left(\mu\right) \textrm{ for } t \in \left(0, \infty\right), \\	
	f_{t}\left(\lambda\right) &> f_{t}\left(\mu\right) \textrm{ for } t \in \left(-\infty, 0\right],
\end{align}
where
\begin{equation}
	f_t\left(\lambda\right) := \left\{ \begin{array}{l}
		\frac{1}{t - 1}\log \left( \sum_{i=1}^{d} \lambda_i^t \right)\\
		\sum_i \log \lambda_i \\
		-\infty
		  \end{array} \right. \begin{array}{l}
		    \textrm{ when } t > 0 \textrm{ and } \lambda_i \neq 0 \textrm{ for all } i \in \{1,...,d\}, \\
			\textrm{ when } t = 0 \textrm{ and } \lambda_i \neq 0 \textrm{ for all } i \in \{1,...,d\}, \\
		    \textrm{ otherwise. }
		  \end{array}
\end{equation}
\end{theorem}

Using the results established in the last section, we can show that catalysis can also occur under PPT operations.

\begin{theorem}\label{t:catcond}
	$\Phi_C\ox\Phi_K \PPTto \Phi_C\ox\rho_{\lambda}$ (where $\rho_\lambda \neq \Phi_K$) for some sufficiently large value of $C$ if and only if $S_{1/2}\left(\lambda\right) < \log K$.
\end{theorem}
\begin{proof}
	For a given value of $C$, the transformation is possible if and only if $T\left(KC;\lambda\ox U_C\right) \leq 1$, where $U_C$ denotes the uniform probability distribution vector with $C$ elements.
	As $C$ tends to infinity, both the upper bound \eref{T-upperBound} and lower bound \eref{T-lowerBound} on $T\left(KC;\lambda\ox U_C\right)$ tend to $2^{S_{1/2}\left(\lambda\right)}/K$. Therefore, if $S_{1/2}\left(\lambda\right) < \log K$ then for some sufficiently large $C$, $T\left(KC;\lambda\ox U_C\right) \leq 1$ so the corresponding transformation is possible. That the condition is necessary follows from Proposition \ref{exactPPTcost} ($S_{1/2}$ cannot increase in \emph{any} PPT pure state transformation and the R\'{e}nyi entropies are additive) and Corollary \ref{c:borderline} (deals with the case where $S_{1/2}$ stays the same).
\end{proof}

In the case where the goal state has Schmidt rank three, the states satisfying
\begin{equation}
	2\left (\sqrt{\lambda^{\uparrow}_3 \lambda^{\uparrow}_2} + \sqrt{\lambda^{\uparrow}_3 \lambda^{\uparrow}_1} + \sqrt{\lambda^{\uparrow}_1 \lambda^{\uparrow}_2 } \right) \leq 1
\end{equation}
are exactly those that can be reached from $\Phi_2$ by PPT operations when maximally entangled catalysts of arbitrarily high rank are available, according to Theorem \ref{t:catcond}. This is a strict superset of those Schmidt rank three states which can be obtained from $\Phi_2$ without a catalyst (see Theorem \ref{t:rank3}). These regions are illustrated on one cell of the simplex of Schmidt coefficient vectors in Fig. \ref{fig:allowed-regions}.
	
\begin{figure}
\includegraphics[scale=0.5]{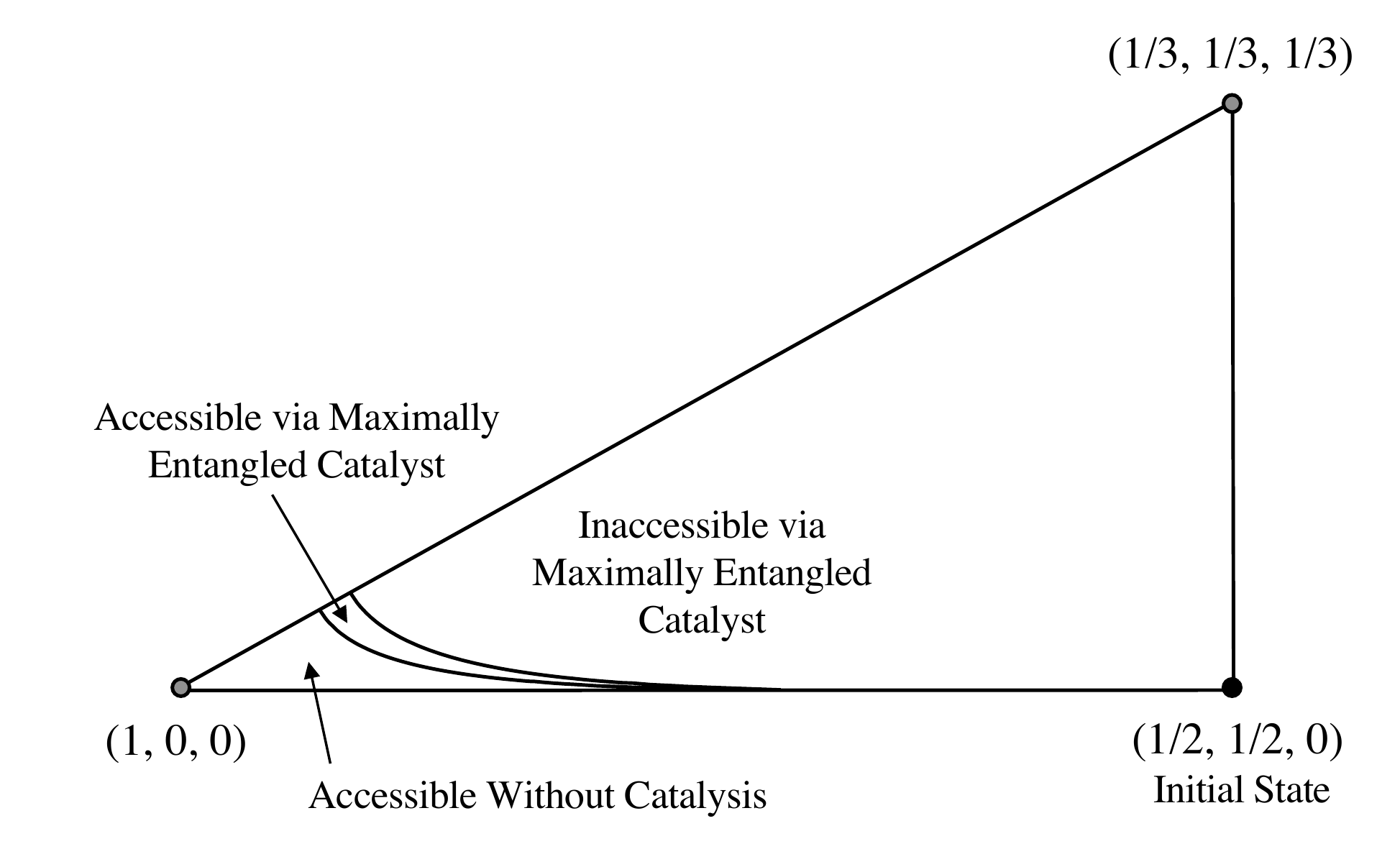}
\caption{Pure states of Schmidt rank three accessible deterministically from a single EPR pair by PPT and by PPT assisted by a maximally entangled catalyst.}\label{fig:allowed-regions}
\end{figure}

\section{Conclusions}\label{sec:conc}

We have provided a necessary condition for the exact preparation of a pure bipartite state from a maximally entangled state by PPT operations in terms of the Schmidt coefficients of the final state. We conjecture that this condition is also sufficient and have shown that this is true when the final state has Schmidt rank three. We have demonstrated that the phenomenon of catalysis occurs in the context of PPT operations. A notable difference from LOCC catalysis is that maximally entangled states can act as catalysts under PPT operations (this is impossible with LOCC). In the case where both the initial state and the catalyst are maximally entangled, we have given a necessary and sufficient condition for the production of a pure state: the R\'{e}nyi entropy at $1/2$ of the Scmidt coefficient vector must decrease, unless the final state is the same as the initial state. To give a direction for future work -- comparing Theorem \ref{t:catcond} and the observations on PPT monotones given in Section \ref{sec:general} with Theorem \ref{t:locccatcond}, we make the following conjecture:

\begin{conjecture}
	When $\lambda^{\uparrow} \neq \mu^{\uparrow}$, there exists a catalyst state $\rho_{\xi}$ such that $\rho_{\xi}\ox\rho_{\lambda} \PPTto \rho_{\xi}\ox\rho_{\mu}$ if and only if
	\begin{equation}
	S_{t}\left(\lambda\right) > S_{t}\left(\mu\right), \forall t \in [1/2,\infty).
	\end{equation}
\end{conjecture}

It would also be desirable to determine the validity of Conjecture \ref{conj:rankIs1}.

\begin{acknowledgments}
WM acknowledges support from the U.K. EPSRC. 
AW is supported by the U.K. EPSRC (project ``QIP IRC'' and an Advanced
Research Fellowship), by a Royal Society Wolfson Merit Award, and
the EC, IP ``QAP''. The Centre for Quantum Technologies is funded by the
Singapore Ministry of Education and the National Research Foundation as
part of the Research Centres of Excellence programme.
We thank Armin Uhlmann and Jens Eisert for early discussions on the topic of this
paper, and Richard Low for his comments on presentation.
\end{acknowledgments}

\vfill\pagebreak

\appendix
\section{Proof of Theorem \ref{t:rank1opt}}
The rank constraint
\begin{equation}
	\rank \left(\dsum_{i\geq j} \mu_{ij} \sigma_{ij}\pt + \dsum_{i > j} t_{ij} \alpha_{ij}\pt\right) = 1
\end{equation}
can be rewritten as
\begin{equation}
\rank\left(\dsum_{i > j} \left(\mu_{ij} + t_{ij}\right)\left(\ketbra{ij}{ij}+\ketbra{ji}{ji}\right)/2 + \dsum_{i}\mu_{ii}\ketbra{ii}{ii} + \dsum_{i > j} \left(\mu_{ij} - t_{ij}\right)\left(\ketbra{ii}{jj}+\ketbra{jj}{ii}\right)/2 \right) = 1.
\end{equation}
The first sum and the remainder of the expression of the operator have disjoint support. The first sum can only have even rank, so to satisfy the condition we must have $t_{ij} = -\mu_{ij}$ (so that it's rank is zero). The remaining terms are then
	$\sum_{i}\mu_{ii}\ketbra{ii}{ii} + \sum_{i > j} \mu_{ij}\left(\ketbra{ii}{jj}+\ketbra{jj}{ii}\right)$
which has rank one if and only if
	$\mu_{ij} = u_i u_j,$
for some $u \in \mathbb{R}^d$.

Making the change of variables $x_i = \sqrt{\lambda_i} u_i$, the rank-constrained version of the optimisation problem \eref{dualSDP} is therefore equivalent to
\begin{equation}
	T_1\left(K; \lambda\right) = \frac{K 2^{S_{1/2}\left(\lambda\right)} - 1}{K^2 - 1} + \max\{\Delta\left(x\right) | x \in S \}
\end{equation}
where
\begin{equation}
	\Delta\left(x\right) = \dsum_{i > j} \frac{x_i x_j}{K-1} - \dsum_{i\geq j} \frac{x_i x_j}{K+1}
\end{equation}
and $S$ is the hypercuboid defined by
\begin{equation}
	|x_i| \leq \sqrt{\lambda_i} \textrm{ for all } i \in \{1,...,d\}.
\end{equation}

\begin{figure}[ht]
\includegraphics[scale=0.7]{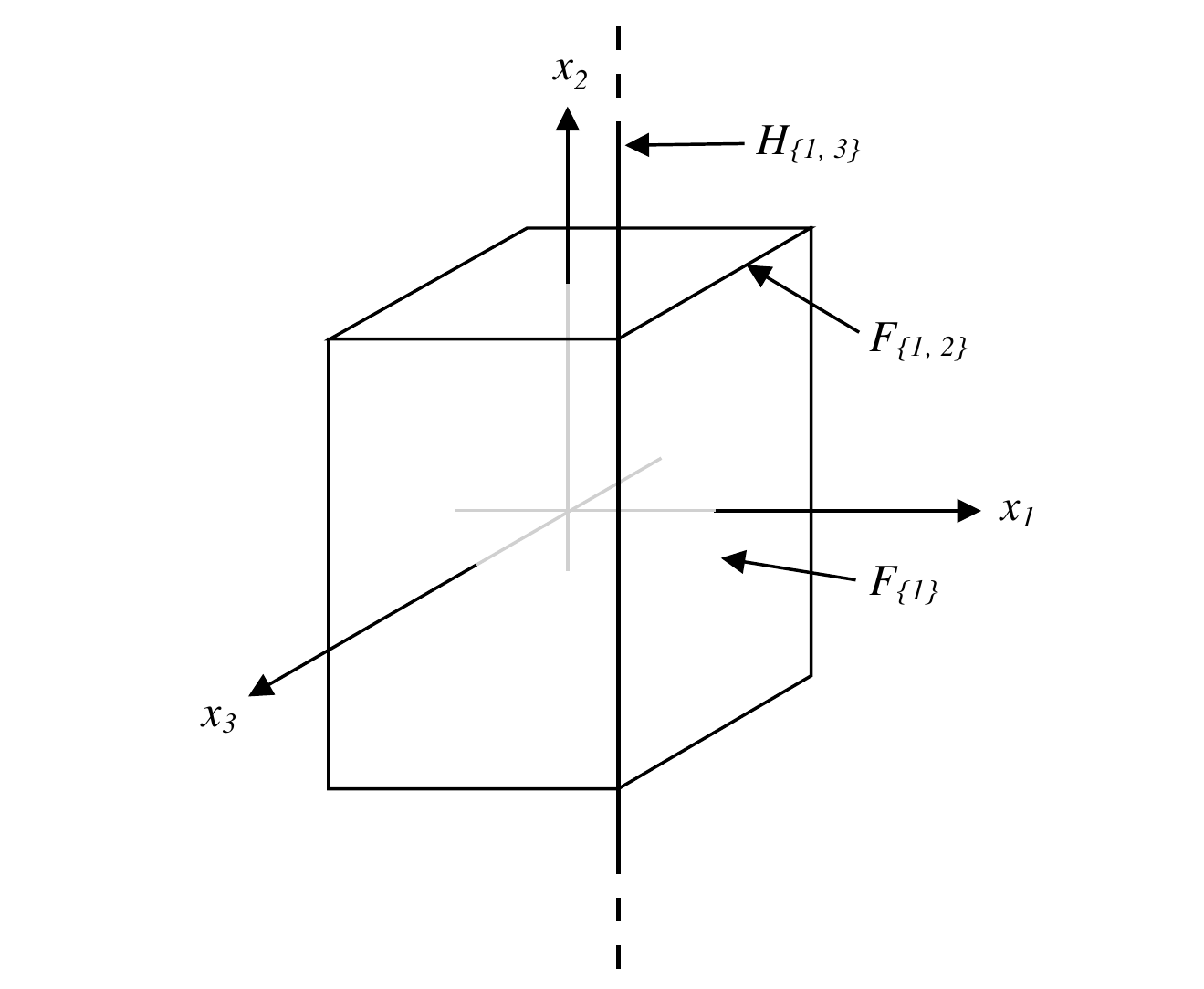}
\caption{Illustration of the notations $F_{C}$ and $H_{C}$.}\label{fig:cuboid}
\end{figure}

In the following we shall use some terminology from convex geometry (`face', `supporting hyperplane'), for definitions see Br\o ndsted \cite{brondsted}, for example.  The function $\Delta$ is differentiable, so its maximum in the set $S$, $\Delta^{\ast}$, is attained by either an extremal point of $S$ or a stationary point of the function on a face of $S$.

If one of the components of $x$ is negative, then changing its sign does not decrease $\Delta\left(x\right)$, so $\Delta^{\ast}$ must be attained by a point on one of the faces of $S$ in the set $\{F_{\mathcal{C}} | \mathcal{C} \subset \{1,...,d\}\}$, where $F_{\mathcal{C}}$ is the intersection of $S$ with the supporting hyperplane $H_{\mathcal{C}}$ defined by $x_i = +\sqrt{\lambda_i}, i \in \mathcal{C}$ (see Fig. \ref{fig:cuboid}).

The restriction of the function $\Delta$ to domain $H_{\mathcal{C}}$ has a single stationary point at
\begin{equation}
	x_j = \left(\dsum_{i\in\mathcal{C}}\sqrt{\lambda_i}\right)/\left(K+|\mathcal{C}|-d\right) \textrm{ for all } j \in \{1, ...,d\}\setminus \mathcal{C},
\end{equation}
where the value of $\Delta$ is equal to
\begin{equation}
	\Delta_{\mathcal{C}} = \frac{K}{\left(K^2-1\right)\left(K+c-d\right)}\left( \left(\dsum_{i \in \mathcal{C}} \sqrt{\lambda_i} \right)^2 - \left(K + c - d\right) \dsum_{i \in \mathcal{C}} \lambda_i \right).
\end{equation}
The restriction of $\Delta$ to face $F_{\mathcal{C}}$ contains a stationary point if and only if
\begin{equation}
	\left(\dsum_{i\in\mathcal{C}}\sqrt{\lambda_i}\right)/\left(K+|\mathcal{C}|-d\right) \leq \min \{\sqrt{\lambda_i} | i \in \{1,...,d\} \setminus \mathcal{C}\}.
\end{equation}

\begin{lemma}
	Suppose we have two subsets of $\{1,...,d\}$, $\mathcal{C}_1$ and $\mathcal{C}_2$, both of size $c$, which differ by one element such that $\mathcal{C}_1 = \mathcal{C}_0 \cup n_1$ and $\mathcal{C}_2 = \mathcal{C}_0 \cup n_2$. If $\Delta$ on $F_{\mathcal{C}_1}$ has a stationary point and $\lambda_{n_1} > \lambda_{n_2}$, then $\Delta$ on $F_{\mathcal{C}_2}$ also has a stationary point and $\Delta_{\mathcal{C}_2} \geq \Delta_{\mathcal{C}_1}$.
\end{lemma}
\begin{proof}
Part (i) is trivial. For part (ii), since $n_2$ is not in $C_1$, the fact that $\Delta$ on $F_{\mathcal{C}_1}$ has a stationary point implies that
\begin{equation}
	\dsum_{i\in\mathcal{C}_0}\sqrt{\lambda_i} \leq \left(K + c - d\right)\sqrt{\lambda_{n_2}} - \sqrt{\lambda_{n_1}},
\end{equation}
and therefore
\begin{align}
\left(K+c-d\right)&\left(K^2-1\right)\left(\Delta_{\mathcal{C}_2} - \Delta_{\mathcal{C}_1}\right)/K\\
 &= \left(K+c-d\right)\left(\lambda_{n_1} - \lambda_{n_2}\right) - \left(\left(\sqrt{\lambda_{n_1}}+\dsum_{i\in\mathcal{C}_0}\sqrt{\lambda_i}\right)^2 -\left(\sqrt{\lambda_{n_2}}+\dsum_{i\in\mathcal{C}_0}\sqrt{\lambda_i}\right)^2\right)\\
 &= \left(\sqrt{\lambda_{n_1}} - \sqrt{\lambda_{n_2}}\right) \left(\left(K + c - d - 1\right)\left(\sqrt{\lambda_{n_1}} + \sqrt{\lambda_{n_2}}\right) - 2\dsum_{i\in\mathcal{C}_0}\sqrt{\lambda_i} \right)\\
 &\geq \left( \sqrt{\lambda_{n_1}}/\sqrt{\lambda_{n_2}} - 1 \right) \left( \sqrt{\lambda_{n_1}} + \sqrt{\lambda_{n_2}} \right) \geq 0.
\end{align}

\end{proof}

If we assume, without loss of generality, that $\lambda_i = \lambda_i^{\uparrow}$, then as result of this lemma, if $\Delta^{\ast}$ is attained by a point in the relative interior of an $m$-dimensional face then it must also be attained by the $m$-dimensional face with $C = \{1,...,d-m\}$. Therefore the $\Delta^{\ast}$ is equal to $\Delta_{\mathcal{C^{\ast}}}$, where $C^{\ast} = \{1,...,c^{\ast}\}$ with $c^{\ast}$ being the smallest value (giving the largest dimensional face) such that ${\Delta}$ has a stationary point on $F_{C^{\ast}}$ or if no face has a stationary point then $c^{\ast} = d$ and the maximum occurs at the extremal point $x = \left(\sqrt{\lambda_1},...,\sqrt{\lambda_d}\right)$. The result follows.

\end{document}